\documentclass[runningheads]{llncs}
\usepackage[english]{babel}
\usepackage{epsfig} 
\usepackage{amssymb}
\usepackage{amsfonts}
\usepackage{amsmath}
\usepackage{graphicx}

\usepackage[algo2e,algoruled,vlined]{algorithm2e}

\def\A{\mathcal{A}}
\def\B{\mathcal{B}}
\def\O{\mathcal{O}}
\def\q#1{\mathtt{#1}}
\def\h{\q{h}}
\def\1{\q{1}}
\def\0{\q{0}}
\def\m{\q{\#}}
\def\j{\q{\_}}
\def\t{\q{@}}

\title{Subset seed automaton}
\titlerunning{Subset seed automaton}
\author{Gregory Kucherov\inst{1}
\and 
Laurent No{\'e}\inst{1}
\and 
Mikhail Roytberg\inst{2}
}
\institute{
LIFL/CNRS/INRIA, B\^at. M3 Cit\'e Scientifique, 59655, Villeneuve d'Ascq cedex, France, \email{\{Gregory.Kucherov,Laurent.Noe\}@lifl.fr}
\and
Institute of Mathematical Problems in Biology, Pushchino, Moscow
Region, 142290, Russia, \email{mroytberg@mail.ru}
}
\authorrunning{G.~Kucherov, L.~No{\'e}, M.~Roytberg}
\date{{\normalsize{\today}}}

\begin{document}
\maketitle
\begin{abstract}
We study the pattern matching automaton introduced in
\cite{KucherovNoeRoytbergJBCB06} for the purpose of seed-based 
similarity search. We show that our definition provides a compact
automaton, much smaller than the one obtained by applying the
Aho-Corasick construction. We study properties of this automaton and
present an efficient implementation of the automaton construction. We
also present some experimental results and show
that this automaton can be successfully applied to more general
situations.  
\end{abstract}

% introduction 
\section{Introduction}

The technique of {\em spaced seeds} 
for similarity search in strings
(sequences) was introduced about five years ago
\cite{BurkhardtKarkkainenFI03,PatternHunter02} and constituted an
important algorithmic development
\cite{BrownLiMaJBCB04,BrownUnpublished07}. Its main applications have
been approximate string matching \cite{BurkhardtKarkkainenFI03} and
local alignment of DNA sequences \cite{PatternHunter02,PatternHunter04,NoeKucherovNAR05} but the 
underlying idea applies also to other algorithmic problems on strings
\cite{FLASH93,TsurWABI06}. 

Since the invention of spaced seeds, different generalizations have
been proposed, such as seeds with match errors \cite{BLASTZ03,SunBuhlerBMCBioinformatics06},
{\em daughter seeds}~\cite{CsurosMaCOCOON05}, {\em indel seeds}
\cite{MakGelfandBensonBioinformatics06}, or 
{\em vector seeds}~\cite{BrejovaBrownVinarJCSS05}. 
In \cite{KucherovNoeRoytbergJBCB06}, we proposed the notion of 
{\em subset seeds} and demonstrated its advantages and its usefulness
for DNA sequence alignment. 
In the formalism of subset seeds, an alignment is viewed as a text
over some alphabet $\A$, and a seed as a pattern over a subset
alphabet $\B\subseteq 2^\A$. The only requirements made is that $\A$
contains a special letter $\1$, $\B$ contains a letter $\m=\{\1\}$,
and every letter of $\B$ contains $\1$ in its set.  The matching relation is
naturally defined: a seed letter $b\in\B$ matches a letter $a\in\A$
iff $a$ belongs to the set $b$. 

For any seed-based similarity search method, including all
above-mentioned types of seeds, an important issue is an 
accurate estimation of the sensitivity of a seed 
with respect to a given probabilistic model
of alignments. For different probabilistic models, this problem has
been studied in
\cite{KeichLiMaTrompDAM04,BuhlerKeichSunRECOMB03,BrejovaBrownVinarJBCB04}.
In \cite{KucherovNoeRoytbergJBCB06} we proposed a general framework
for this problem that allows one to compute the 
seed sensitivity for different definitions of seed and different
alignment models. This approach is based on a finite
automata representation of the set of target alignments and
the set of alignments matched by a seed, as well as on a
representation of the probabilistic model of alignments as a
finite-state transducer. 

A key ingredient of the approach of \cite{KucherovNoeRoytbergJBCB06}
is a finite automaton that recognizes the set of alignments matched
(or {\em hit}) by a given subset seed. We call this automaton a 
{\em subset seed automaton}. The size (number of states) of the subset
seed automaton is crucial for the efficiency of the whole algorithm of
\cite{KucherovNoeRoytbergJBCB06}. Note that the algorithm of
\cite{BuhlerKeichSunRECOMB03} is also based on an automaton
construction, namely on the Aho-Corasick automaton implied by the
well-known string matching algorithm. 

Besides its application to the seeding technique for
similarity search and string matching, constructing an efficient
subset seed automaton is an interesting problem in its own, as it
provides a solution to a variant of the {\em subset matching problem}
studied in literature
\cite{ColeHariharanSODA99,HolubSmythWangJDA06,RahmanIliopoulosMouchardWALCOM07}. 

In this paper, we study properties of the subset seed automaton and 
present an efficient implementation of its construction. More
specifically, we obtain the following results:
\begin{itemize}
\item we present a construction of subset seed automaton that has
  $\O(w2^{s-w})$ states, compared to $\O(w|\A|^{s-w})$ implied by
  the Aho-Corasick construction, where $s$ and $w$ are respectively
  the {\em span} and the {\em weight} of the seed defined in the next Section,
\item we further motivate our construction by showing that for some
  seeds, our construction gives the minimal automaton, 
\item we prove that our automaton is {\em always} smaller than the one
  obtained by the Aho-Corasick construction; we provide experimental
  data that confirm that for $|\A|=2$, our automaton is on average
  about 1.3 times bigger than the minimal one, while the 
  Aho-Corasick automaton is about 2.5 times bigger. For $|\A|=3$ the
  difference is much more substantial: while our automaton is still
  about 1.3 times bigger than the minimal one, the Aho-Corasick
  automaton turns out to be about 17 times bigger,
\item we provide an efficient algorithm that implements the
  construction of the automaton such that each transition is computed
  in constant time,
\item we show that our construction can be applied to the case of
  multiple seeds and to the general subset matching problem. 
\end{itemize}

The presented automaton construction is implemented in full generality
in {\sc Hedera} software package
(\url{http//bioinfo.lifl.fr/yass/hedera.php}) and has been applied to
the design of efficient seeds for the comparison of genomic
sequences.

\section{Subset seed matching}
\label{section:SubsetSeeds}

The goal of seeds is to specify short string patterns that, if
shared by two strings, have best chances to belong to a larger similarity
region common to the two strings. To formalize this, a similarity region
is modeled by an alignment between two strings. Usually one considers
{\em gapless alignments} that, in the simplest case, are viewed as
sequences of matches and mismatches 
and are easily specified by binary strings $\{\0,\1\}^*$, where $\1$
is interpreted as ``match'' and $\0$ as ``mismatch''. A {\em spaced seed}
is a string over binary alphabet $\{\m,\j\}$. 
The length of $\pi$ is called its {\em span} and the number of $\m$ is
called its {\em weight}. A spaced seed
$\pi\in\{\m,\j\}^s$ {\em matches} (or {\em hits}) an alignment
$A\in\{\0,\1\}^*$ at a position $p$ if for all $i\in [1..s]$,
$\pi[i]=\m$ implies $A[p+i-1]=\1$. 

In~\cite{KucherovNoeRoytbergJBCB06}, we proposed a generalization of
this basic framework, based on the idea to
distinguish between different types of mismatches in the alignments. This
leads to representing both alignments and seeds as words over larger
alphabets. In the general case, consider an alignment alphabet $\A$ of
arbitrary size. We always assume that $\A$
contains a symbol $\1$, interpreted as ``match''. 
A {\em subset seed} is defined as a word over a {\em seed alphabet}
$\B$, such that 
\begin{itemize}
\item each letter $b\in\B$ denotes a subset of 
$\A$ that contains $\1$ ($b\in 2^{\A}\setminus 2^{\A\setminus\{\1\}}$), 
\item  $\B$ contains a letter $\m$ that denotes subset $\{\1\}$.
\end{itemize}
As before, $s$ is called the {\em span} of $\pi$, and 
the {\em $\m$-weight} of $\pi$ is the number of $\m$ in $\pi$. 
A subset seed $\pi\in\B^s$ {\em matches} an alignment
$A\in\A^*$ at a position $p$ iff for all 
$i \in [1..s]$, $A[p+i-1] \in \pi[i]$.  

\begin{example}\label{example:subset-transition}
For DNA sequences over the alphabet $\{\mathtt{A,C,G,T}\}$, in
\cite{NoeKucherovBMCBioinformatics04} we considered
the alignment alphabet $\A = \{\1,\h,\0\}$ representing respectively a
match, a transition mismatch ($\mathtt{A}\leftrightarrow\mathtt{G}$,
$\mathtt{C}\leftrightarrow\mathtt{T}$), or a transversion mismatch
(other mismatch).
In this case, the appropriate seed alphabet is $\B = \{\m,\t,\j\}$ 
corresponding
respectively to subsets $\{\1\}$, $\{\1,\h\}$, and $\{\1,\h,\0\}$.
Thus, seed $\pi = {\m\t\j\m}$ matches alignment $A = {\1\0\h\1\h\1\1\0\1}$
at positions $4$ and $6$. The span of $\pi$ is $4$, and the
$\m$-weight of $\pi$ is 2.
\end{example}

One can view the problem of finding seed occurrences in an alignment
as a special string matching problem. In particular, it can be
considered as a special case of {\em subset matching}
\cite{ColeHariharanSODA99} where the text is composed of
individual characters. It is also an instance of the problem of
matching in indeterminate (degenerate) strings
\cite{HolubSmythWangJDA06,RahmanIliopoulosMouchardWALCOM07}.
Therefore, an efficient automaton construction that we present in the
following sections applies directly to these instances of string
matching. One can also freely use the string matching terminology by
replacing words ``seed'' and ``alignment'' by ``pattern'' and ``text''
respectively. 

\section{Subset Seed Automaton}
\label{section:SubsetSeedAutomaton}

% -- Fix our notation for the rest of the description
Let us fix an alignment alphabet $\A$, a seed alphabet $\B$, and a
seed $\pi = \pi_1 \ldots \pi_{s} \in \B^*$ of span $s$ and $\m$-weight
$w$. Denote $r = s - w$ and let $R_{\pi}$, $|R_{\pi}|= r$, be the set
of all non-$\m$ positions in $\pi$.
Throughout the paper, we identify each position $z\in R_{\pi}$ with
the corresponding prefix $\pi_{1..z}=\pi_1\ldots \pi_z$ of $\pi$, and we 
interchangeably
regard elements of $R_{\pi}$ as positions or as prefixes of $\pi$.

% -- The state definition
We now define an automaton
$S_{\pi} = <Q, q_0, Q_F, \A,\psi : Q\times\A\to Q>$, $q_0\in Q$,
$Q_F\subseteq Q$,
that recognizes the set of all alignments matched by $\pi$.
The states $Q$ are defined as pairs $\langle X,t \rangle$ such that
$X =\{x_1,\ldots,x_k\}\subseteq R_{\pi}$, $t \in [ 0\ldots s ]$,
$\max\{X\}+t\leq s$. The
automaton maintains the following invariant condition.
Suppose that $S_{\pi}$ has read a prefix $a_1\ldots a_p$ of an
alignment $A$ and has come to a state $\langle X,t \rangle$.
Then $t$ is the length of the longest suffix of $a_1\ldots a_p$ of the
form $\1^i$, $i\leq s$,
and $X$ contains all positions $x_i\in R_\pi$
such that prefix $\pi_{1..x_i}$ matches a suffix of
$a_{1}\cdots a_{p-t}$.
 
\def\PiMotif#1#2{
  \put(#1,#2){
     \put(0,0){$\mbox{\tt \#}$}
     \put(1,0){$\mbox{\tt  @}$}
     \put(2,0){$\mbox{\tt \#}$}
     \put(3,0){$\mbox{\tt \_}$}
     \put(4,0){$\mbox{\tt \#}$}
     \put(5,0){$\mbox{\tt \#}$}
     \put(6,0){$\mbox{\tt \_}$}
     \put(7,0){$\mbox{\tt \#}$}
     \put(8,0){$\mbox{\tt \#}$}
     \put(9,0){$\mbox{\tt \#}$}
  }
}
\def\StrMotif#1#2{
  \put(#1,#2){
     \put(0,0){$\mbox{\tt 1}$}
     \put(1,0){$\mbox{\tt 1}$}
     \put(2,0){$\mbox{\tt 1}$}
     \put(3,0){$\mbox{\tt h}$}
     \put(4,0){$\mbox{\tt 1}$}
     \put(5,0){$\mbox{\tt 0}$}
     \put(6,0){$\mbox{\tt 1}$}
     \put(7,0){$\mbox{\tt 1}$}
     \put(8,0){$\mbox{\tt h}$}
     \put(9,0){$\mbox{\tt 1}$}
     \put(10,0){$\mbox{\tt 1}$} 
     \put(11,0){$\mbox{\tt }$}
  }
}
\def\PisevenPref#1#2{
  \put(#1,#2){
     \put(0,0){$\mbox{\tt \#}$}
     \put(1,0){$\mbox{\tt  @}$}
     \put(2,0){$\mbox{\tt \#}$}
     \put(3,0){$\mbox{\tt \_}$}
     \put(4,0){$\mbox{\tt \#}$}
     \put(5,0){$\mbox{\tt \#}$}
     \put(6,0){$\mbox{\tt \_}$}
  }
}
\def\PifourPref#1#2{
  \put(#1,#2){
     \put(0,0){$\mbox{\tt \#}$}
     \put(1,0){$\mbox{\tt  @}$}
     \put(2,0){$\mbox{\tt \#}$}
     \put(3,0){$\mbox{\tt \_}$}
  }
}
\def\PitwoPref#1#2{
  \put(#1,#2){
     \put(0,0){$\mbox{\tt \#}$}
     \put(1,0){$\mbox{\tt  @}$}
  }
}
% -- example to add here
\begin{figure*}[htb]\center
  \begin{picture}(100,30)(-2,-10)\noindent\centering\setlength{\unitlength}{5pt}
    \put(-19, 1){$(a)$}
    \put(-15.3, 1){$\pi = $}\PiMotif{-12}{ 1}% seed motif
    \put(-19,-3){$(b)$}
    \put(-15.5,-3){$A = $}\StrMotif{-12}{-3}% text 
    \put(8,1){$(c)$}
    \put(22.7, 4){$a_9$}
    \put(23.5, 3.5){\line(0,-1){0.5}}
    % [
    \put(24.1, 3.5){\line(0,-1){0.5}}
    \put(25.7, 3.5){\line(0,-1){0.5}}
    \put(24.1, 3.5){\line(1, 0){1.6}}
    \put(24.7, 4){$t$}
    % ]
    \StrMotif{15}{2}
    \put(11, 0.1){$\pi_{1..7}=$}\PisevenPref{17}{0} %prefixes
    \put(14,-1.4){$\pi_{1..4}=$}\PifourPref{20}{-1.5}
    \put(16,-2.9){$\pi_{1..2}=$}\PitwoPref{22}{-3}
  \end{picture}%
\caption{\label{figure:Statedef} Illustration to Example~\ref{example:subset-state}}
\vspace{-5mm}
\end{figure*}
\begin{example}
\label{example:subset-state}
In the framework of Example~\ref{example:subset-transition}, 
consider a seed $\pi$ and an alignment prefix $A=a_1\ldots a_{p}$ 
of length $p = 11$ given in Figure~\ref{figure:Statedef}(a) and (b) respectively. The length $t$ of the last run of
$\1$'s of $A$ is $2$. The last non-$\1$ letter of $A$ is $a_9 = \h$.
The set $R_\pi$ of non-$\m$ positions of $\pi$ is $\{2,4,7\}$ and 
$\pi$ has 3 prefixes belonging to $R_\pi$ (Figure~\ref{figure:Statedef}(c)). 
Prefixes $\pi_{1..2}$ and $\pi_{1..7}$ do match suffixes of 
$a_1 a_2\ldots a_9$, but prefix $\pi_{1..4}$ does not.
Thus, the state of the automaton after reading $a_1 a_2\ldots a_{11}$ is
$\langle \{2,7\}, 2 \rangle$.
\end{example}

The initial state $q_0$ of $S_{\pi}$ is the state $ \langle \emptyset,
0 \rangle$.
Final states $Q_F$ of $S_{\pi}$ are all states $q = \langle X,t \rangle$, where
$max\{X\} + t = s$. All final states are merged into one state $\langle\rangle$. 

% -- Transition function $\psi$
The transition function $\psi(q,a)$ is defined as follows.
If $q$ is a final state, then $\forall a \in \A$, $\psi(q,a) = q$. 
If $q = \langle X,t \rangle$ is a non-final state, then
\begin{itemize}
\item if $a = \1$ then $\psi(q,a) = \langle X,t+1 \rangle$, 
\item otherwise $\psi(q,a) =  \langle X_{U} \cup X_{V},0 \rangle$ with
  \begin{itemize}
  \item $X_{U} = \{ x         \;\: | \;\; x \leq t + 1
    \mathrm{~and~} a \in \pi_{x} \}$
  \item $X_{V} = \{ x + t + 1  \;\: | \;\; x \in X  \mathrm{~and~}
    a \in \pi_{x+t+1} \}$
  \end{itemize}
\end{itemize}
% -- figure
\vspace{-1.0cm}%
\begin{figure*}[htb]\center
  \includegraphics[width=12cm]{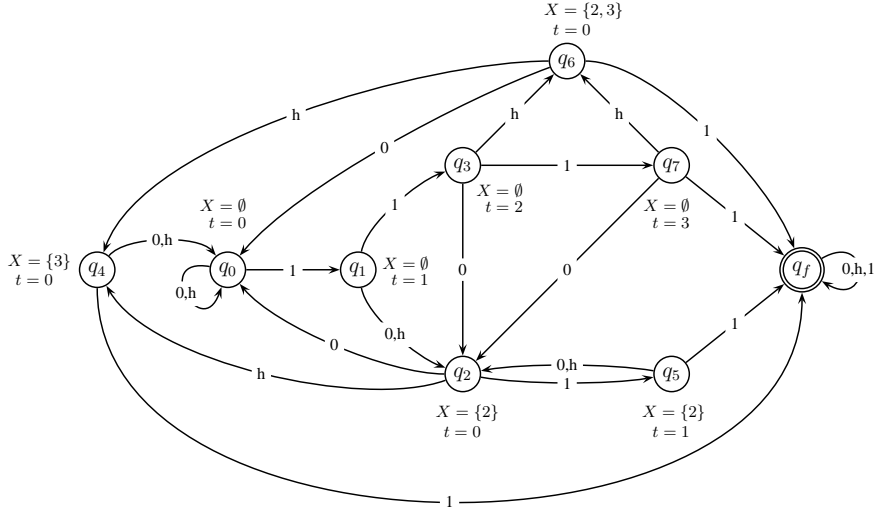}\vspace{-0.5cm}%
  \caption{\label{figure:ss-automaton} Illustration to Example~\ref{example:subset-automaton}}
\end{figure*}% 
\vspace{-1.0cm}%
% -- example --
\begin{example}\label{example:subset-automaton}
Still in the framework of Example~\ref{example:subset-transition},
consider seed $\pi = \m\j\t\m$. Then the set $R_\pi$ is $\{2,3\}$.
Possible non-final states $\langle X,t \rangle$ of $S_\pi$ are states
$\langle \emptyset,0 \rangle$, $\langle \emptyset,1 \rangle$,
$\langle \emptyset,2 \rangle$, $\langle \emptyset,3 \rangle$,
$\langle \{2\},0 \rangle$, $\langle \{2\},1 \rangle$, $\langle
\{3\},0 \rangle$, $\langle \{2,3\},0 \rangle$.  All
these states are reachable in
$S_\pi$. Figure~\ref{figure:ss-automaton} shows the resulting
automaton.
\end{example}

We now study main properties of automaton $S_\pi$. 
% -- Automaton exactly recognizes the seed language
\begin{lemma}\label{lemma:Gpi}
The automaton $S_{\pi}$ accepts all alignments $A\in\A^*$ matched by $\pi$.
\end{lemma}
\begin{proof}
It can be verified by induction that the invariant condition on the states
$\langle X,t \rangle\in Q$ is preserved by the transition function $\psi$. 
The final state verifies $max\{X\} + t = s$ which implies
that at the first time $S_{\pi}$ gets into the final state, $\pi$ matches a suffix of $a_1\ldots a_p$. \qed
\end{proof}
%
% -- Number of states
\begin{lemma}\label{lemma:GnbStates}
The number of states of the automaton $S_{\pi}$ is no more than
$(w+1)2^{r}$, where $w$ is the $\m$-weight of $\pi$. 
\end{lemma}
\begin{proof}
Assume that $R_\pi=\{z_1,z_2,\ldots,z_r\}$
and $z_1 < z_2 \cdots < z_r$. Let $Q_i$ be the set of non-final states
$\langle X,t \rangle$ with $max\{X\} = z_i$. For states $q = \langle
X,t \rangle \in Q_i$ there
are $2^{i-1}$ possible values of $X$ and $s-z_i$ possible values of 
$t$ between $0$ and $s-z_i-1$, as $max\{X\}+t\leq s-1$. 

Thus,\vspace{-0.5cm}%
{\small
\begin{eqnarray}\label{equation:nbStates}
                       |Q_i|   \;\leq\;                2^{i-1} (s-z_i)
                       & \leq & 2^{i-1} (s-i), \mbox{ and} \\
      \sum_{i=1}^{r}   |Q_i|   \;\leq\; \sum_{i=1}^{r} 2^{i-1} (s-i)   & =    &(s-r+1) 2^r - s - 1.
\end{eqnarray}}
Besides states $Q_i$, $Q$ contains $s$ states $\langle \emptyset,t \rangle$
($t\in [0..s-1]$) and one final state. 
Thus, $|Q| \leq (s-r+1) 2^r = (w+1) 2^r$. \qed
\end{proof}

Note that if $\pi$ starts with $\m$, which is always the case for
spaced seeds, then $X_i \geq i+1$, $i\in[1..r]$, and the
bound of 
(\ref{equation:nbStates}) rewrites to $2^{i-1} (s-i-1)$. 
This results in the same $w 2^r$ bound on number of states 
as the one for
the Aho-Corasick automaton proposed in \cite{BuhlerKeichSunRECOMB03}
for spaced seeds (see also Lemma~\ref{surjection} below). 

The next Lemma shows that the construction of automaton $S_\pi$ is
optimal in the sense that no two states can be merged in general.
\begin{lemma}\label{lemma:nbStatesWorstCase}
Let $\A=\{\0,\1\}$ and $\B=\{\m,\j\}$, where $\m=\{\1\}$ and
$\j=\{\0,\1\}$. 
Consider a seed $\pi=\m\j\cdots \j\m$ with $r$ letters '$\j$' between
two $\m$'s. Then the automaton $S_\pi$ is reduced, that is
\begin{itemize}
\item[(i)] each of its states $q$ is reachable, and
\item[(ii)] any two non-final states $q', q''$ are not equivalent.
\end{itemize}
\end{lemma}
\begin{proof}

(i) Let $q = \langle X, t \rangle$ be a non-final state of the
automaton $S_\pi$, and let $X =\{x_1,\ldots, x_k\}$ with  $ x_1 < \cdots <  x_k$. 
Let $A=a_1\ldots a_{x_k}\in \{\0, \1\}^*$ be an alignment of length
$x_k$  defined as follows: $a_{p}=\1$ if, for some $i \in [1 ..k]$,  $p = x_k
- x_i + 1$, and $a_p=\0$ otherwise. 
Note that $1 \notin X $ and thus $a_{x_k}=0$. Thus
$\psi (\langle \emptyset, 0 \rangle, A) = \langle X,0 \rangle$ and finally
$\psi (\langle \emptyset,0 \rangle, A \cdot 1^t) = q$. 

(ii) For a set $X =\{x_1,\ldots, x_k\}$ and an
integer $t$, denote $X \oplus t = \{x_1+t,\ldots, x_k+t\}$. 
Let $q' = \langle X', t' \rangle$ and $q'' = \langle X'', t'' \rangle$ be non-final states of
$S_\pi$. 
If $max \{ X' \} + t' > max \{ X'' \} + t''$, then let $d =
(r+2) - (max \{ X' \} + t')$. Obviously, 
$\psi(q' ,\1^d)$ is a final state, and $\psi(q'' ,\1^d)$ is not.

Now assume that $max \{ X' \} + t' = max \{ X'' \} + t''$.
Let $g = max \{ v | (v  \in X' \oplus t' \mbox{~ and~} v \notin X''\oplus t'') 
\mbox{~or~}    (v \in X'' \oplus t''   \mbox{~and~} v \notin X' \oplus t') \}$.
By symmetry, assume that the maximum is reached on the first
condition, i.e. $g=x'_i+t'$ for some $x'_i\in X'$. Let 
$d = (r+1) - g$ and consider word $\0^d \1$. It is easy to see
that $\psi(q',\0^d \1)$ is a final state. We claim that
$\psi(q'',\0^d \1)$ is not. 
To see this, observe that none of the seed prefixes corresponding to $x\in X''$
with $x+t''>x'_i+t'$ can lead to the final state on $\0^d \1$,
due to the last $\m$ symbol of $\pi$. The details are left to the
reader. 
\qed
\end{proof}

%--- Surjective mapping
Another interesting property of $S_{\pi}$ is the existence of a
surjective mapping from the states of the Aho-Corasick automaton onto
reachable states of $S_\pi$. 
This mapping proves that even if $S_\pi$ is not always minimized, it
has {\em always} a smaller number of states than the Aho-Corasick automaton. 
Here, by the Aho-Corasick (AC) automaton, we mean the automaton with the
states corresponding to nodes of the trie built according to the
classical Aho-Corasick construction \cite{AhoCorasick74} from the set of
all instances of the seed $\pi$. More precisely, given a seed $\pi$ of
span $s$, the set of states of the AC-automaton is $Q_{AC}=\{A\in\A^*|\,|A|\leq
s \mbox{ and } A \mbox{ is matched by prefix } \pi_{1..|A|} \}$.
The transition $\psi(A,a)$ for $A\in Q_{AC}$, $a\in\A$ yields the
{\em longest} $A'\in Q_{AC}$ which is a suffix of $Aa$. 
We assume that all final states are merged into a single sink state. 

%--- Lemma Surjective mapping
\begin{lemma}
\label{surjection}
Consider an alignment alphabet $\A$, a seed alphabet $\B$ and a seed
$\pi\in\B^s$ of span $s$.  
There exists a surjective mapping $f:Q_{AC} \rightarrow Q$ from
the set of states of the Aho-Corasick automaton to the set of
reachable states of the subset seed automaton $S_{\pi}$. 
\end{lemma}
\begin{proof}
We first define the mapping $f$. Consider a state $A\in Q_{AC}$,
$|A|=p<s$, where $A$ is matched by $\pi_{1..p}$. Decompose
$A=A'{\1}^t$, where the last letter of $A'$ is not $\1$. If $A'$ is
empty, define $f(A)=\langle \emptyset,t \rangle$. Otherwise,
$\pi_{1..p-t}$ matches $A'$ and $\pi[p-t]\neq\m$. Let $X$ be a set of
positions that contains $p-t$ together with all positions $i<p-t$
such that $\pi_{1..i}$ matches a suffix of $A'$. Define
$f(A)=\langle X,t \rangle$. It is easy to see that $\langle X,t
\rangle\in Q$, that $\langle X,t \rangle$
exists in $S_{\pi}$ and is reachable by string $A$. 

Now show that for every reachable state $\langle X,t \rangle\in Q$ of $S_{\pi}$
there exists $A\in Q_{AC}$ such that $f(A)=\langle X,t \rangle$. Consider a string
$C\in\A^*$ that gets $S_{\pi}$ to the state $\langle X,t \rangle$. Then $C=C'{\1}^t$
and the last letter of $C'$ is not $\1$. If $X$ is empty then define
$A=\1^t$. If $X$ is not empty, then consider the suffix $A'$ of $C'$ of
length $x=\max\{X\}$ and define $A=A'\1^t$. Since $\pi_{1..x}$ matches
$A'$, and $x+t\leq s$, then $\pi_{1..x+t}$ matches $A$ and therefore
$A\in Q_{AC}$. It is easy to see that $f(A)=\langle X,t \rangle$. 
\qed
\end{proof}

Observe that the mapping of Lemma~\ref{surjection} is actually a
morphism from the Aho-Corasick automaton to $S_{\pi}$. 

Table~\ref{table:Avg1}
shows experimentally estimated average sizes of the Aho-Corasick
automaton, subset seed automaton, and minimal automaton. The two
tables correspond respectively to the binary alphabet (spaced seeds)
and ternary alphabet (see Example~\ref{example:subset-transition}). 
For Aho-Corasick and subset seed automata, the ratio to the average
size of the minimal automaton is shown. Each line corresponds to a
seed weight ($\m$-weight for $|\A|=3$). In each case, 10000
random seeds of different span have been generated to estimate the
average. 

\begin{table*}[htb]%
  \begin{center}% 
    {\scriptsize%
      \begin{tabular}{c|cc|cc|cc}%
        {\bf $|\A| = 2$} &
        \multicolumn{2}{|c}{Aho-Corasick}&
        \multicolumn{2}{|c}{$S_\pi$}&
        \multicolumn{1}{|c}{Minimized}\\
        $w$     &
        $avg.$  &
        $ratio$ &
        $avg.$  &
        $ratio$ &
        $avg.$  \\
        \hline
        \hline
        9  & 130.98 & 2.46 & 67.03 & 1.260 & 53.18\\
        10 & 140.28 & 2.51 & 70.27 & 1.255 & 55.98\\
        11 & 150.16 & 2.55 & 73.99 & 1.254 & 58.99\\
        12 & 159.26 & 2.57 & 77.39 & 1.248 & 62.00\\
        13 & 168.19 & 2.59 & 80.92 & 1.246 & 64.92\\
      \end{tabular}%
      ~~~
      \begin{tabular}{c|cc|cc|cc}
        {\bf  $|\A| = 3$} &
        \multicolumn{2}{|c}{Aho-Corasick}&
        \multicolumn{2}{|c}{$S_\pi$}&
        \multicolumn{1}{|c}{Minimized}\\
        $w$    &
        $avg.$ &
        $ratio$&
        $avg.$ &
        $ratio$&
        $avg.$ \\
        \hline
        \hline
        9  & 1103.5 & 16.46 &  86.71 & 1.293 & 67.05\\
        10 & 1187.7 & 16.91 &  90.67 & 1.291 & 70.25\\
        11 & 1265.3 & 17.18 &  95.05 & 1.291 & 73.65\\
        12 & 1346.1 & 17.50 &  98.99 & 1.287 & 76.90\\
        13 & 1419.3 & 17.67 & 103.10 & 1.284 & 80.31\\
      \end{tabular}
    }
    \caption{\it\label{table:Avg1} Average number of states of Aho-Corasick, $S_\pi$ and minimal automaton}%
  \end{center}%
  \vspace{-5mm}
\end{table*}%

%------------
% -- Proof Subset seed automaton w2^r build procedure
\section{Subset seed automaton implementation}
\label{section:SubsetSeedAutomatonImplementation}
As in section~\ref{section:SubsetSeedAutomaton}, 
consider a subset seed $\pi$ of $\#$-weight $w$ and span $s$, and let 
$r = s-w$ be the number of non-$\#$ positions.
A straightforward generation of the transition table of the automaton
$S_\pi$ can be performed in time $\O(r \cdot w \cdot 2^{r} \cdot |\A|
)$.
In this section, we show that $S_\pi$ can be constructed in time
proportional to its size, which is bounded by $(w+1)2^r$, according
to Lemma~\ref{lemma:GnbStates}. In practice, however, the number of
states is usually much smaller.

% -- breadth-first principle
The algorithm generates the states of the automaton incrementally by
traversing them in the breadth-first manner. Transitions $\psi(\langle
X,t \rangle,a)$ are computed using previously computed transitions
$\psi(\langle X',t \rangle,a)$. A tricky part of the algorithm
corresponds to the case where state $\psi(\langle X,t \rangle,a)$ has
already been created before and should be retrieved.

The whole construction of the automaton is given in
Algorithm~\ref{algorithm:subsetseedautomaton}. We now describe it in
more details.

% -- some notations
Let $R_\pi = \{z_1,\ldots,z_r\}$ and $z_{1}< z_{2}\cdots <
z_r$. Consider $X \subseteq R_\pi$. To retrieve the maximal element of
$X$, the algorithm maintains a function $k(X)$ defined by

\[
k(X) = \max\{i | z_i \in X\},\ k(\emptyset)=0.
\]

Let $q=\langle X,t \rangle$ be a non-final and reachable state of
$S_\pi$, $X=\{x_{1},\ldots,x_{i}\}\subseteq R_\pi$ and $x_{1} <
x_{2} \cdots < x_{i}$. We define $X' = X \setminus \{z_{k(X)}\} =
\{x_{1},\ldots,x_{i-1}\}$ and $q'=\langle X',t \rangle$. The following
lemma holds.

% -- reachability proof
\begin{lemma}\label{lemma:TransitionFunctionPrecedence}
If $q=\langle X,t \rangle$ is reachable, then $q'=\langle X',t
\rangle$ is reachable and has been processed before in a breadth-first
computation of $S_\pi$.
\end{lemma}
\begin{proof}
  First prove that $\langle X',t \rangle$ is reachable.
  If $\langle X,t \rangle$ is reachable, then  $\langle X,0 \rangle$
  is reachable due to the definition of transition function for $t >
  0$.  Thus, there is a word $A$ of length $x_i = z_{k(X)}$ such that
  $\forall j \in [1..r]$,  $z_j \in X$ iff the seed suffix
  $\pi_{1..z_j}$ matches the word suffix $A_{x_i-z_j+1} \cdots
  A_{x_i}$.
  Define $A'$ to be the suffix of $A$ of length $x_{i-1}=z_{k(X')}$ and observe
  that reading $A'$ gets the automaton to
  the state $\langle X',0 \rangle$, and then reading $A'\cdot 1^t$
  leads to the state $\langle X',t \rangle$.
  Finally, as $|A' \cdot 1^t| <  |A \cdot 1^t| $, then the breadth-first
  traversal of states of $A_\pi$ always processes state $\langle X',t \rangle$
  before $\langle X,t \rangle$. \qed
\end{proof}%

% -- fail function
To retrieve $X'$ from $X$, the algorithm maintains a function
${\mbox{\sc Fail}}(q)$, similar to the {\em failure} function of the
Aho-Corasick automaton, such that ${\mbox{\sc Fail}}(\langle X,t
\rangle)=\langle X',t \rangle$ for $X\neq\emptyset$, and
${\mbox{\sc Fail}}(\langle \emptyset,t \rangle)=\langle
\emptyset,max\{t-1,0\} \rangle$. 

% -- psi value computed
We now explain how values $\psi(q,a)$ are computed by
Algorithm~\ref{algorithm:subsetseedautomaton}. 
Note first that if $a = \1$, state $\psi(q,a)=\langle X,t+1 \rangle$ 
can be computed in constant time (part a. of Algorithm~\ref{algorithm:subsetseedautomaton}).
Moreover, since this is the only way to reach state $\langle X,t+1
\rangle$, it is created and added once to the set of states.

Assume now that $a \neq \1$. To compute $\psi(q,a)=\langle
Y,0\rangle$, we retrieve state $q'={\mbox{\sc Fail}}(q)=\langle X',t\rangle$ and then
retrieve $\psi(q',a)=\langle Y',0\rangle$. Note that this is
well-defined as by Lemma~\ref{lemma:TransitionFunctionPrecedence},
$q'$ has been processed before $q$. 

Observe now that since $X'$
and $X$ differ by only one seed prefix $\pi_{1..z_{k(X)}}$
the only possible difference between $Y$ and $Y'$ can be the prefix
$\pi_{1..z_{k(X)}+t+1}$ depending on whether $\pi_{z_{k(X)}+t+1}$ matches
$a$ or not.  As $a\neq\1$, this is equivalent to testing whether
$(z_{k(X)}+t+1) \in R_{\pi}$ and $\pi_{z_{k(X)}+t+1}$ matches $a$. This
information can be precomputed for different values $k(X)$ and $t$. 

For every $a\neq \1$, we define
\begin{eqnarray*}
  V (k,t,a) & = & \begin{cases}  
    \{z_k + t + 1\} & \mbox{~if~} z_k + t + 1\in R_{\pi}
    \mbox{~and~} \; \pi_{z_k + t + 1}  \mbox{~matches~} a, \\   
    \emptyset & \mbox{~otherwise.} \\
  \end{cases}\\ 
\end{eqnarray*}  
Thus, $Y = Y' \cup V(k(X),t,a)$ (part c. of Algorithm~\ref{algorithm:subsetseedautomaton}). Function $V (k,t,a)$
can be precomputed in time and space $\mathcal{O}(|\A| \cdot r \cdot
s)$.

% -- details on V(k,t,a), existence of states and RevMaxFail
Note that if $V(k,t,a)$ is empty, then $\langle Y,0 \rangle$ is equal to
an already created state $\langle Y',0 \rangle$ and no new state needs
to be created in this case (part e. of
Algorithm~\ref{algorithm:subsetseedautomaton}).

If  $V(k,t,a)$ is not empty, we need to find out if $\langle Y,0
\rangle$ has already been created or not and if it has, we need to
retrieve it. To do that, we need an additional construction.
For each state $q' = \langle X',t \rangle$, we maintain another
function ${\mbox{\sc RevMaxFail}}(q')$, that gives
the {\em last created} state $q = \langle X,t \rangle$ such that
$X \backslash z_{k(X)} = X'$ (part d. of
Algorithm~\ref{algorithm:subsetseedautomaton}).
Since the state generation is breadth-first, new states $\langle X,t
\rangle$ are created in a non-decreasing order of the quantity $(z_{k(X)}
+ t)$. Therefore, among all states $\langle X,t \rangle$ such that
${\mbox{\sc Fail}}(\langle X,t \rangle)=\langle X',t \rangle$,
${\mbox{\sc RevMaxFail}}(\langle X',t \rangle)$ returns the one with
the largest $z_{k(X)}$. 

Now, observe that if $V(k,t,a)$ is not empty,
i.e. $Y=Y'\cup\{z_{k(X)}+t+1\}$, then
${\mbox{\sc Fail}}(\langle Y,0\rangle)=\langle Y',0 \rangle$. Since
state $\langle Y,0\rangle$ has the maximal possible current value
$z_{k(Y)}+0=z_{k(X)}+t+1$, by the above remark, we conclude that if
$\langle Y,0\rangle$ has already been created, then
${\mbox{\sc RevMaxFail}}(\langle Y',0 \rangle)=\langle Y,0
\rangle$. This allows us to check if this is indeed the case and to
retrieve the state $\langle Y,0\rangle$ if it exists (part d. of
Algorithm~\ref{algorithm:subsetseedautomaton}).

% -- details on U(t,a) 
The generation of states $\langle X,t \rangle$ with $X=\emptyset$
represents a special case (part b. of
Algorithm~\ref{algorithm:subsetseedautomaton}). Here another
precomputed function is used:
\begin{eqnarray*}
U(t,a) & = & \cup \{  x  | x \leq t + 1 \mathrm{~and~} a
\mathrm{~matches~} \pi_{x}\}
\end{eqnarray*}
$U(t,a)$ gives the set of seed prefixes that match the
word $1^{t} \cdot a$. In this case, checking if resulting states have
been already added is done in a similar way to
$V(k,t,a)$. Details are left out.

{\small
\begin{algorithm2e}[H]
  \label{algorithm:subsetseedautomaton}
  \caption{computation of $S_{\pi}$}
  \KwData{a seed $\pi = \pi_1\pi_2 \ldots \pi_s$}
  \KwResult{an automaton $S_\pi = \langle Q,q_0,q_F,\mathcal{A},\psi \rangle$}
  $q_F \leftarrow createstate(\langle\rangle)$; $q_0 \leftarrow createstate(\langle \emptyset,0 \rangle)$\;
  \emph{/* process the first level of states to set  ${\mbox{\sc Fail}}$ and
    ${\mbox{\sc RevMaxFail}}$ */ }\\
  \For{$a \in \mathcal{A}$}{
    \eIf{$a \in \pi_1$}{
      \eIf{$a = \1$}{
        $\langle Y,t_y \rangle \leftarrow  \langle \emptyset,1 \rangle$\;
      }
      {
        $\langle Y,t_y \rangle \leftarrow  \langle \{1\},0 \rangle$\;
      }
      \eIf{$z_{k(Y)} + t_y \geq s$}{
        $q_y \leftarrow q_F$\;
      }{
        $q_y \leftarrow createstate(\langle Y,t_y \rangle)$\;
        ${\mbox{\sc Fail}}(q_y) \leftarrow q_0$; ${\mbox{\sc RevMaxFail}}(q_0) \leftarrow  q_y$\;
        $push(Queue,q_y)$\;
      }
    }{
      $q_y \leftarrow q_0$\;
    }
    $\psi(q_0,a) \leftarrow q_y$\;
  }
  ~\\
  \emph{/* breadth-first processing */ }\\
  \While{$Queue \neq \emptyset$}{
    $q  : \langle X,t_X    \rangle \leftarrow  pop(Queue)$\;
    $q' \leftarrow  {\mbox{\sc Fail}}(q)$\;
    \For{$a \in \mathcal{A}$}{
      \emph{/* compute $\psi(\langle X,t_X \rangle,a) = \langle Y,t_y \rangle$ */ }\\
      $q'_{y} : \langle Y',t'_{y} \rangle \leftarrow  \psi(q',a)$\;            
      \eIf{$a = \1$}{
        $Y   \leftarrow   X$\;
        \nlset{a} $t_y \leftarrow t_X + 1$\;
      }{
        \eIf{$X = \emptyset$}{
          \nlset{b} $Y \leftarrow  U(t_X,a)$\;
        }{
          \nlset{c} $Y \leftarrow  Y' \cup V(k(X),t_X,a)$\;
        }
        $t_y \leftarrow 0$\;
      }
      ~\\
      \emph{ /* create a new state unless it already exists or it is final */ }\\
      $q_{rev} : \langle Y_{rev},t_{rev} \rangle \leftarrow {\mbox{\sc RevMaxFail}}(q'_{y})$\;
      \uIf{$defined(q_{rev})$ {\bf and } $t_y = t_{rev}$ {\bf and } $Y = Y_{rev}$}{
              \nlset{d} $q_y \leftarrow q_{rev}$\;
      }
      \uElseIf{$t_y = t'_{y}$ {\bf and } $Y = Y'$}{
        \nlset{e}   $q_y \leftarrow q'_{y}$\;
      }
      \Else{
        \eIf{$z_{k(Y)} + t_y \geq s$}{
          $q_y \leftarrow q_{F}$\;
        }{
          $q_y \leftarrow createstate(\langle Y,t_y \rangle)$\;
          ${\mbox{\sc Fail}}(q_y) \leftarrow q'_{y}$;  ${\mbox{\sc RevMaxFail}}(q'_{y}) \leftarrow q_y$\;
          $push(Queue,q_y)$\;
        }
      }
      $\psi(q,a) \leftarrow q_y$\;
    }
  }
\end{algorithm2e}
}%small

We summarize the results of this section with the following Lemma.   
\begin{lemma}\label{lemma:TransitionFunctionConstantTime}
After a preprocessing of seed $\pi$ within time $\O(|\A|\cdot s^2)$, the
automaton $S_{\pi}$ can be constructed by incrementally generating all
reachable states so that every transition
$\psi(q,a)$ is computed in constant time. 
\end{lemma}

\section{Possible extensions}
% -- multiple seeds generalization
An important remark is that the automaton defined in this paper can be
easily generalized to the 
case of multiple seeds. For seeds $\pi^1, \ldots, \pi^k$, a state of
the automaton recognizing the alignments matched by one of the seeds
would be a tuple $<X_1,\ldots,X_k,t>$, where $X_1,\ldots,X_k$ contain
the set of respective prefixes, similarly to the construction of the
paper. Interestingly, Lemma~\ref{surjection} still holds for the case
of multiple seeds. This means that although the size of the union of
individual seed automata could potentially grow as the product of
sizes, it actually does not, as it is bounded by the size of the
Aho-Corasick automaton which grows additively with respect to subsets
of underlying words. 
In practice, our
automaton is still substantially smaller than the Aho-Corasick
automaton, as illustrated by Table~\ref{table:Avg2}. Similar to
Table~\ref{table:Avg1}, 10000 random seed pairs have been generated
here in each case to estimate the average size. 

\begin{table*}[htb]%
  \vspace{-0.5cm}
  \begin{center}% 
    {\scriptsize%
      \begin{tabular}{c|cc|cc|cc}%
        {\bf $|\A| = 2$} &
        \multicolumn{2}{|c}{Aho-Corasick}&
        \multicolumn{2}{|c}{$S_\pi$}&
        \multicolumn{1}{|c}{Minimized}\\
        $w$    &
        $avg.$ &
        $ratio$&
        $avg.$ &
        $ratio$&
        $avg.$ \\
        \hline
        \hline
        9  & 224.49 & 2.01 & 122.82 & 1.10 & 111.43\\
        10 & 243.32 & 2.07 & 129.68 & 1.10 & 117.71\\
        11 & 264.04 & 2.11 & 137.78 & 1.10 & 125.02\\
        12 & 282.51 & 2.15 & 144.97 & 1.10 & 131.68\\
        13 & 300.59 & 2.18 & 151.59 & 1.10 & 137.74\\
      \end{tabular}%
      ~~~
      \begin{tabular}{c|cc|cc|cc}
        {\bf  $|\A| = 3$} &
        \multicolumn{2}{|c}{Aho-Corasick}&
        \multicolumn{2}{|c}{$S_\pi$}&
        \multicolumn{1}{|c}{Minimized}\\
        $w$    &
        $avg.$ &
        $ratio$&
        $avg.$ &
        $ratio$&
        $avg.$ \\
        \hline
        \hline
         9 & 2130.6 & 12.09 & 201.69 & 1.15 & 176.27\\
        10 & 2297.8 & 12.53 & 209.75 & 1.14 & 183.40\\
        11 & 2456.5 & 12.86 & 218.27 & 1.14 & 191.04\\
        12 & 2600.6 & 13.14 & 226.14 & 1.14 & 198.00\\
        13 & 2778.0 & 13.39 & 236.62 & 1.14 & 207.51\\
      \end{tabular}
    }
    
    \caption{\it\label{table:Avg2} Average number of states of
      Aho-Corasick, $S_\pi$ and minimized automata for the case of two seeds}%
  \end{center}%
  \vspace{-0.8cm}
\end{table*}%

% -- the IUPAC example
Another interesting observation is that the construction of a matching 
automaton where each state is associated with a set of ``compatible''
prefixes of the pattern is a general one and can be applied to
the general problem of subset matching
\cite{ColeHariharanSODA99,AmirPoratLewensteinSODA01,HolubSmythWangJDA06,RahmanIliopoulosMouchardWALCOM07}. 
Recall that in subset matching, a pattern is composed of subsets of 
alphabet letters. This is the case, for example, with IUPAC
genomic motifs, such as motif $\mathtt{ANDGR}$ representing the
subset motif ${\scriptstyle A[ACGT][AGT]G[AG]}$. Note that the
text can also be composed of subset letters, with two possible
matching interpretations \cite{RahmanIliopoulosMouchardWALCOM07}: a seed letter
$b$ matches a text letter $a$ either if $a\subseteq b$ or if
$a\cap b\neq\emptyset$. 

Interestingly, the automaton construction of this paper still
applies to these cases with minor modifications due to the absence of text
letter $\1$ matched by any seed letter. With this modification, the
automaton construction algorithm of
Section~\ref{section:SubsetSeedAutomatonImplementation} still
applies. As a test case, we applied it to subset motif
${\scriptstyle [GA][GA]GGGNNNNAN[CT]ATGNN[AT]NNNNN[CTG]}$ 
mentioned in \cite{RahmanIliopoulosMouchardWALCOM07} as a motif describing the
translation initiation site in the {\em E.coli} genome. 
For a regular 4-letters genomic text, the automaton obtained 
with our approach has only 138 states, while the minimal
automaton has 126 states.
For a text composed of 15 subsets of 4 letters and the inclusion
matching relation, our automaton contains 139 states, compared to 127
states of the minimal automaton. However, in the case of intersection
matching relation, the automaton size increases drastically: it
contains 87617 states compared to the 10482 states of the minimal
automaton.

\bibliographystyle{splncs}
\bibliography{main}
\end{document}